\documentclass[11pt]{article}

\usepackage{amssymb} 
\usepackage{amsmath}     
\usepackage{amsthm}
\usepackage{todonotes}
\usepackage{algorithmic,algorithm}

\usepackage{a4wide}
\usepackage{graphicx}
\usepackage{hyperref}

\newtheorem{thm}{Theorem}
\newtheorem{lem}{Lemma}

\newtheorem{pro}{Property}

\newenvironment{keyword}{\par{\noindent\bf Keywords:}}

\begin{document}

\title{Robust recoverable and two-stage selection problems}

\author{Adam Kasperski\thanks{Corresponding author}\\
   {\small \textit{Faculty of Computer Science and Management,}}
  {\small \textit{Wroc{\l}aw University of Science and Technology,}}\\
  {\small \textit{Wybrze{\.z}e Wyspia{\'n}skiego 27,}}
  {\small \textit{50-370 Wroc{\l}aw, Poland,}}
  {\small \textit{adam.kasperski@pwr.edu.pl}}
  \and
  Pawe{\l} Zieli{\'n}ski\\
    {\small \textit{Faculty of Fundamental Problems of Technology,}}
  {\small \textit{Wroc{\l}aw University of Science and Technology,}}\\
  {\small \textit{Wybrze{\.z}e Wyspia{\'n}skiego 27,}}
  {\small \textit{50-370 Wroc{\l}aw, Poland,}}
  {\small \textit{pawel.zielinski@pwr.edu.pl}}} 
  
  \date{}
    
\maketitle

\begin{abstract}
In this paper the following selection problem is discussed. A set of $n$ items is given and we wish to choose  a subset of exactly $p$ items of the minimum total cost.  This problem is a special case of  0-1 knapsack in which all the item weights are equal to~1. Its deterministic version has an $O(n)$-time algorithm, which consists in choosing $p$ items of the smallest costs.   
 In this paper it is assumed that the item costs are uncertain.
Two robust models, namely two-stage and recoverable ones,  under discrete and interval uncertainty representations, are discussed. Several positive and negative complexity results for both of them are provided.
\end{abstract}

\begin{keyword}
robust optimization; computational complexity; approximation algorithms; selection problem
\end{keyword}

\section{Introduction}
\label{sintro}

In this paper we wish to investigate the following \textsc{Selection} problem. Let $E=\{e_1,\dots,e_n\}$ be a set of items. Each item $e\in E$ has a nonnegative cost $c_e$ and we wish to choose a subset $X\subseteq E$ of exactly $p$ items of the minimum total cost, $f(X)=\sum_{e \in X} c_e$, where  $p\in [n]=\{1,\dots,n\}$. This problem has a trivial $O(n)$-time algorithm which works as follows. We first determine in $O(n)$-time the $p$th smallest item cost, say $c_{(p)}$ 
(see, e.g.,~\cite{CO90}), and then choose $p$ items from $E$ whose costs are not greater than $c_{(p)}$. \textsc{Selection}  is a special, polynomially solvable version of the 0-1 knapsack problem, in which all the items have unit weights. It is also a special case of some other discrete optimization problems such as minimum assignment or a single machine scheduling problem with the weighted number of late jobs criterion (see~\cite{KKZ13} for more details). It  can be seen as a basic resource
allocation problem~\cite{IK88}. It is also a matroidal problem, as the set of feasible solutions is composed of all bases of an uniform matroid~\cite{PS98}.

Suppose  that the item costs are uncertain and we are given a \emph{scenario set} $\mathcal{U}$, which contains all possible vectors of the item costs, called \emph{scenarios}. We thus only know that one cost scenario $S\in\mathcal{U}$ will occur, but we do not know which one before a solution is  computed. The cost of item $e\in E$ under scenario $S\in \mathcal{U}$ is denoted as $c_e^S$ and we assume that $c_e^S\geq 0$.
No additional information for scenario set $\mathcal{U}$, such as a probability distribution, is provided. Two methods of defining scenario sets are popular in the existing literature (see, e.g,~\cite{KY97}). In the \emph{discrete uncertainty representation}, $\mathcal{U}^D=\{S_1,\dots,S_K\}$ contains $K>1$ explicitly listed scenarios. In the \emph{interval uncertainty representation}, for each item $e\in E$ an interval $[\underline{c}_e,\overline{c}_e]$ of its possible costs is specified and $\mathcal{U}^I=\Pi_{e\in E} [\underline{c}_e,\overline{c}_e]$ is the Cartesian product of these intervals.  The cost of solution $X$ depends now on scenario $S\in \mathcal{U}$, $\mathcal{U}\in\{\mathcal{U}^D,\mathcal{U}^I\}$, and will be denoted as $f(X,S)=\sum_{e\in X} c_e^S$.
In order to choose a solution two robust criteria, namely the \emph{min-max} and \emph{min-max regret} can be applied, which lead to the following two optimization problems:
$$\textsc{Min-Max Selection}: \min_{X\in \Phi} \max_{S\in \mathcal{U}} f(X,S),$$
$$\textsc{Min-Max Regret Selection}: \min_{X\in \Phi} \max_{S\in \mathcal{U}} (f(X,S)-f^*(S)),$$
where $\Phi=\{X\subseteq E: |X|=p\}$ is the set of all feasible solutions and $f^*(S)$ is the cost of an optimal solution under scenario $S$. The quantity $f(X,S)-f^*(S)$ is called a \emph{regret} of solution $X$ under scenario $S$.
Both robust versions of the \textsc{Selection} problem have attracted a considerable attention in the recent literature. It turns out that their complexity depends on the way in which scenario set $\mathcal{U}$ is defined. 
 It has been shown in~\cite{A01} that, under the discrete uncertainty representation, \textsc{Min-Max Regret Selection} is NP-hard even for two scenarios.
 Repeating a similar argument as the one used in~\cite{A01}  gives the result that
  \textsc{Min-Max Selection} remains NP-hard even for two scenarios.
  Both problems become strongly NP-hard when the number of scenarios is a part of input~\cite{KKZ13}.
 Furthermore, in this case \textsc{Min-Max Selection} is  also hard to approximate within any constant factor~\cite{KKZ13}. Several approximation algorithms for \textsc{Min-Max Selection} have been recently proposed in~\cite{KZ09a, KKZ13, D13}. The best known, designed in~\cite{D13}, has an approximation ratio of $O(\log K/ \log \log K)$.
 For the interval uncertainty representation both robust problems are polynomially solvable. The min-max version is trivially reduced to a deterministic counterpart, as it is enough to solve the deterministic problem for scenario $(\overline{c}_e)_{e\in E}$. On the other hand, \textsc{Min-Max Regret Selection} is more involved and some polynomial time algorithms for this problem have been constructed in~\cite{A01,C04}. The best known algorithm with running time $O(n\cdot\min\{p,n-p\})$ has been shown in~\cite{C04}.

Many real world problems arising in operations research and optimization have a two-stage nature. Namely, a complete or a partial solution is determined in the first stage and can be then modified or completed in the second stage. Typically, the costs in the first stage are known while the costs in the second stage are uncertain. This uncertainty is also modeled by providing a scenario set $\mathcal{U}\in\{\mathcal{U}^D,\mathcal{U}^I\}$, which contains all possible vectors of the second stage costs. If no additional information with $\mathcal{U}$ is provided, then the robust criteria can be applied to choose a solution. In this paper we investigate two well known concepts, namely \emph{robust two-stage} and \emph{robust recoverable} ones and apply them to the \textsc{Selection} problem. In the robust two-stage model, a partial solution is formed in the first stage and completed optimally when a true scenario reveals. In the robust recoverable model a complete solution must be formed in the first stage, but it can be modified to some extent after a true scenario occurs.
A key
difference between the models is that the robust two-stage model
pays for the items selected only once, while the recoverable model
pays for items chosen in both stages with the possibility of
replacing a set of items from the first to the second stage,
controlled by the recovery parameter.

 Both models have been discussed in the existing literature for various problems. In particular, the robust two-stage versions of 
the covering~\cite{DGRS05},
the matching~\cite{KMU08} and the minimum spanning tree~\cite{KZ11} problems have been investigated.  The two-stage model has been also considered in the stochastic setting, i.e. when a probability distribution in scenario set is available. Namely, it has been applied to the minimum spanning tree~\cite{EGMS10}, the 0-1 knapsack~\cite{KL11, KO14}, the matching~\cite{KMU08} and the maximum weighted forest~\cite{AALL15} problems.
The robust recoverable approach has been applied to linear programming~\cite{LLMS09}, some network problems~\cite{B11, B12, NO13}, the 0-1 knapsack~\cite{BKK11}, and recently to the traveling salesperson~\cite{CG15b} and the minimum spanning tree~\cite{HKZ16a} problems.

\paragraph{Our results}

 In Section~\ref{secrec} we will investigate the robust recoverable model. We will show that it is strongly NP-hard and not at all approximable, when the number of scenarios is a part of input. A major part of Section~\ref{secrec} is devoted to constructing a polynomial $O((p-k+1)n^2)$ algorithm for the interval uncertainty representation, where $k$ is the recovery parameter.
 In Section~\ref{sec2st} we will study the robust two-stage model. We will prove that it is NP-hard for two second-stage cost scenarios. Furthermore, when the number of scenarios is a part of input, then the problem becomes strongly NP-hard and it has an approximability lower bound of~$\Omega(\log n)$.
 For scenario set $\mathcal{U}^D$, we will construct a randomized algorithm which returns an $O(\log K + \log n)$-approximate solution with high probability.  If $K=poly(n)$, then the randomized algorithm gives the best approximation up to a constant factor.
 We will also show that for the interval uncertainty representation the robust two-stage model is solvable in $O(n)$ time.

\section{Problems formulation}
\label{pform}

Before  we show the formulations of the problems we recall some notations and introduce new ones. Let us fix $p\in [n]$ and define
\begin{itemize}
\item $\Phi=\{X\subseteq E: |X|=p\}$, 
\item $\Phi_1=\{X\subseteq E: |X|\leq p\}$,
\item $\Phi_X=\{Y\subseteq E\setminus X: |Y|=p-|X|\}$, 
\item  $\Phi_X^k=\{Y\subseteq E: |Y|=p, |Y\setminus X|\leq k\}$, $k\in [p] \cup \{0\}$,
\item $C_e$ is the deterministic, first-stage cost of item $e\in E$,
\item $c^S_e$ is the second-state cost of item $e\in E$ under scenario $S\in \mathcal{U}$, where $\mathcal{U}\in \{\mathcal{U}^D, \mathcal{U}^I\}$.
\item $f(X,S)=\sum_{e\in X} c_e^S$, for any subset $X\subseteq E$.
\end{itemize}

 We now define the two-stage model as follows. In the first stage  we choose a subset $X\in \Phi_1$ of the items, i.e. such that $|X|\leq p$, and we add additional $p-|X|$ items to $X$, after observing which scenario in $\mathcal{U}$ has occurred. The cost of $X$ under scenario $S$ is defined as
$$f_1(X,S)=\sum_{e \in X}C_e + \min_{Y\in \Phi_X} f(Y,S).$$
In the \emph{robust two-stage selection} problem we seek $X\in \Phi_1$, which minimizes the maximum cost over all scenarios, i.e.
$$\textsc{Two-Stage Selection}:\; opt_1=\min_{X\in \Phi_1} \max_{S\in\mathcal{U}} f_1(X,S).$$

We now define the robust recoverable model. In the first stage we must choose a complete solution $X\in \Phi$, i.e. such that $|X|=p$. In the second stage additional costs occur for the selected items. However, a limited recovery action is allowed, which consists in replacing at most $k$ items in $X$ with some other items from $E\setminus X$, where $k\in [p] \cup\{0\}$ is a given \emph{recovery parameter}. The cost of $X$ under scenario $S$ is defined as follows:
$$f_2(X,S)=\sum_{e \in X} C_e + \min_{Y\in \Phi_X^k} f(Y,S).$$

In the \emph{robust recoverable selection} problem we wish to find a solution $X\in \Phi$, which minimizes the maximum cost over all scenarios, i.e. 
$$\textsc{Recoverable Selection}: \; opt_2=\min_{X\in \Phi} \max_{S\in \mathcal{U}} f_2(X,S).$$

\section{Robust recoverable selection}
\label{secrec}

In this section we deal with the \textsc{Recoverable Selection} problem.
Consider first  the discrete uncertainty representation, i.e. the problem with scenario set $\mathcal{U}^D$. It is easy to observe that when all the first stage costs are equal to zero and $k=0$, then \textsc{Recoverable Selection} is equivalent to  \textsc{Min-Max Selection} with scenario set $\mathcal{U}^D$. It follows from the fact that the solution formed in the first stage cannot be changed in the second stage and thus $\Phi_X^0=\{X\}$ and  $f_2(X,S)=f(X,S)$. Hence, we have an immediate consequence of the results obtained in~\cite{A01, KKZ13}.  Namely,
under scenario set $\mathcal{U}^D$, the \textsc{Recoverable Selection} problem is  weakly NP-hard when $K=2$. Furthermore, it becomes strongly NP-hard and hard to approximate within any constant factor when $K$ is a part of  input. We now strengthen this result if  the recovery parameter $k$ is allowed to be positive.

\begin{thm}
\label{thmrecd2}
If $K$ is a part of input and $k>0$,
then  \textsc{Recoverable Selection}
 with scenario set~$\mathcal{U}^{D}$
 is strongly NP-hard and not at all approximable unless P=NP.
\end{thm}
\begin{proof}
The reduction is similar to that in~\cite{KZ09a}.
Consider an instance of the strongly NP-complete \textsc{E3-SAT} problem~\cite{GJ79}, 
 in which we are given a set of boolean variables $x_1,\dots,x_n$ and a collection of clauses ${A}_1,\dots A_m$, where each clause is disjunction of exactly three literals (variables or their negations). 
 We ask if there is a truth assignment to the variables which satisfies all the clauses.
We now construct the corresponding instance of the \textsc{Recoverable Selection} 
problem
as follows.
We associate with each clause $A_i=(l_1^i\vee l_2^i \vee l_3^i)$
three items $e_1^i$, $e_2^i$, $e_3^i$
corresponding to three literals in $A_i$. We also create one recovery item~$r$. This gives us the item set $E$, such that
$|E|=3m+1$.  
The first-stage cost of the recovery item~$r$ is set to~$n$ and  the first-stage costs of the remaining items are set to zero.
The scenario set $\mathcal{U}^{D}$ is  formed  as follows.
For each pair of  items  $e_u^s$ and $e_w^t$,
that corresponds to contradictory
 literals  $l_u^s$ and $l_w^t$, i.e. $l_u^s=\overline{l}_w^t$,
we create scenario $S$ such that under this scenario the costs of $e_u^s$ and $e_w^t$ are set to~1 and the costs
 of all the remaining items are set to~0.
For each clause $A_i$, $i\in [m]$, 
we form scenario in which  the costs of 
$e_1^i$, $e_2^i$, $e_3^i$ are 
set to~1 and the rest of items have zero costs. 
We complete the reduction by setting $p=m$ and $k=1$.
 An example of the reduction is shown in Table~\ref{tabrsku}. We will show that the answer to \textsc{E3-SAT} is yes if and only if $opt_2=0$ in the corresponding \textsc{Recoverable Selection} problem.
 
 \begin{table}[h]
	  \centering
	  \caption{The scenarios set for $x_1,\dots, x_3$ and the clauses 
	  $(x_1\vee \overline{x}_2\vee \overline{x}_3)$,
           $(\overline{x}_1\vee x_2 \vee x_3)$, $(x_1 \vee x_2 \vee x_3)$, $p=3$ and $k=1$.}\label{tabrsku}
    \begin{tabular}{c|c|cccccc|ccc}
      $E$&$C_e$&$S_1$&$S_2$&$S_3$&$S_4$&$S_5$&$S_6$&
      $S_7$&$S_8$&$S_9$\\
      \hline
    $e_1^1$&0&1&0&0&0&0&0&1&0&0 \\ 
    $e_2^1$&0&0&1&1&0&0&0&1&0&0\\ 
    $e_3^1$&0&0&0&0&1&1&0&1&0&0\\
    \hline
    $e_1^2$&0&1&0&0&0&0&1& 0&1 &0  \\
    $e_2^2$&0&0&1&0&0&0&0&0&1&0 \\
    $e_3^2$&0&0&0&0&1&0&0&0&1&0 \\
    \hline
    $e_1^3$&0&0&0&0&0&0&1&0&0&1\\
    $e_2^3$&0&0&0&1&0&0&0&0&0&1\\
    $e_3^3$&0&0&0&0&0&1&0&0&0&1\\
    \hline
    $r$&$n$&0&0&0&0&0&0&0&0&0\\
    \end{tabular}   
\end{table}

It is easy to check that
if the answer to \textsc{E3-SAT} is yes, then there is 
a selection $X$  out of $E$ such that $|X|=p$,
containing 
items that do not correspond to contradictory literals
and literals from the same clauses. We form $X$ by choosing exactly one item out of $e^i_1, e^i_2, e^i_3$ for each $i\in [m]$, which corresponds to a true literal in the truth assignment. Thus
the costs  of $X$ under each scenario $S\in \mathcal{U}^{D}$ are at most~1. Hence
we can decrease them to zero
by replacing an item from~$X$ with the cost of~1 under any scenario 
with recovery item $r$ and 
$\max_{S\in \mathcal{U}^{D}} f_2(X,S)= 0$.
 On the other hand, if the answer to \textsc{3-SAT}  is no, then all selections 
$X$, $|X|=p$,  must contain at least two items corresponding to contradictory literals or at least two literals from the same clause.
Note that  the recovery action under each scenario is limited to one.
So,  $\max_{S\in \mathcal{U}^{D}} f_2(X,S)\geq 1$. 
Accordingly,  \textsc{Recoverable Selection}
 with scenario set~$\mathcal{U}^{D}$
 is strongly NP-hard and not at all approximable unless P=NP. 
 \end{proof}

In the remaining part of this section we will provide a polynomial algorithm for the interval uncertainty representation.
The
 \textsc{Recoverable Selection} problem   with scenario set~$\mathcal{U}^{I}$ can be rewritten as follows:
\begin{align} 
\min_{X\in\Phi}
\left\{\sum_{e\in X} C_e  +    \max_{S\in \mathcal{U}^{I}}\min_{Y\in \Phi^{k}_{X}}\sum_{e\in Y}c^{S}_e \right\}=
\min_{X\in\Phi}
\left\{\sum_{e\in X} C_e  +    \min_{Y\in \Phi^{k}_{X}}\sum_{e\in Y}\overline{c}_e \right\}.
\label{irs}
\end{align}
 In problem~(\ref{irs}) we need to find a pair of solutions $X\in \Phi$ and $Y\in\Phi^k_X$. Since $|X|=|Y|=p$, the problem~(\ref{irs}) is equivalent to: 
\begin{equation}
 \begin{array}{ll}
 \min &\sum_{e\in X}C_e + \sum_{e\in Y}\overline{c}_e\\
\text{s.t. } & |X \cap Y| \geq p-k \\
       & X,Y\in \Phi
 \end{array}
 \label{mpsp}
\end{equation}
 In the following, for notation convenience, we will  use $C_i$ to denote $C_{e_i}$ and $\overline{c}_i$ to denote $\overline{c}_{e_i}$ for $i\in [n]$.
Let us introduce 0-1 variables $x_i$, $y_i$ and $z_i$, $i\in [n]$, 
that indicate the chosen parts of $X$, $Y$ and $X\cap Y$, respectively, $X,Y\in \Phi$.  
Problem~(\ref{mpsp}) can be then represented as the following IP model:
\begin{equation}
\label{mip0}
\begin{array}{rlllllll}
	opt_2=\min &\sum_{i=1}^n C_i x_i + \sum_{i=1}^n \overline{c}_i y_i + \sum_{i=1}^n (C_i+\overline{c}_i)z_i \\
\text{s.t. } & 	x_1+\dots + x_n+z_1+\dots + z_n=p  \\
	&y_1+\dots+y_n + z_1+\dots + z_n =p \\
	&z_1+\dots +z_n \geq p-k \\
	&x_i+z_i \leq 1 & i\in [n]  \\
	&y_i+z_i\leq 1 & i\in [n] \\
	& x_i, y_i,z_i\in \{0,1\} & i\in [n] \\
\end{array}
\end{equation}
Let $x_i, y_i, z_i$, $i\in[n]$ be an optimal solution to~(\ref{mip0}). Then $e_i\in X$ if $x_i=1$ or $z_i=1$, and $e_i\in Y$ if $y_i=1$ or $z_i=1$.
The following theorem shows the unimodularity property of the constraint matrix of~(\ref{mip0}).
\begin{thm}
\label{thmtu1}
	The constraint matrix of~(\ref{mip0}) is totally unimodular.
\end{thm}
\begin{proof}
We will use the following
Ghouira-Houri's characterization of  totally unimodular matrices~\cite{GH62}. An $m \times n$ integral matrix is totally unimodular if and only if  each set $R\subseteq [m]$ can be partitioned into two disjoint sets $R_1$ and $R_2$  such that
\begin{equation}
	\sum_{i\in R_1} a_{ij}-\sum_{i\in R_2} a_{ij}\in \{-1,0,1\}, \;\; j\in [n].
\end{equation}
This criterion can alternatively be stated as follows.  An $m \times n$ integral matrix is totally unimodular if and only if for any subset of rows $R=\{r_1,\dots,r_l\}\subseteq [m]$ there exists a coloring of rows of $R$, with 1 or -1, i.e. $l(r_i)\in\{-1,1\}$, $r_i\in R$, such that the weighted sum of every column (while restricting the sum to rows in $R$) is $-1$, $0$ or $1$.
The constraint matrix of~(\ref{mip0}) is shown in Table~\ref{tcmmip}.
\begin{table}[h]
	\caption{The constraint matrix of~(\ref{mip0}).} \label{tcmmip}
\setlength{\tabcolsep}{3pt}
\centering
\begin{tabular}{c|ccccc|ccccc|ccccc|}
& $x_1$ & $x_2$ & $x_3$ & \dots & $x_n$ & $y_1$ & $y_2$ & $y_3$  & \dots & $y_n$  & $z_1$ & $z_2$ & $z_3$  & \dots & $z_n$\\ \hline
$a_1$: & 1 & 1 & 1  & $\dots$ & 1 & 0 & 0 & 0 & $\dots$ & 0& 1 & 1 & 1 &  $\dots$ & 1 \\
$a_2$: & 0 & 0 & 0  & $\dots$ & 0 & 1 & 1 & 1  & $\dots$ & 1 & 1 & 1 & 1 & $\dots$ & 1\\
$a_3$: & 0 & 0 & 0  & $\dots$ & 0 & 0 & 0 & 0  & $\dots$ & 0 & 1 & 1 & 1 &  $\dots$ & 1\\ \hline
$b_1$: &1 &  0 & 0  & $\dots$ & 0 & 0 & 0 & 0  & $\dots$ & 0 & 1 &0 & 0 &  $\dots$ & 0\\
$b_2$: &0 &  1 & 0  & $\dots$ & 0 & 0 & 0 & 0  & $\dots$ & 0 & 0 &1 & 0 &  $\dots$ & 0\\
$b_3$: &0 &  0 & 1 & $\dots$ & 0 & 0 & 0 & 0  & $\dots$ & 0 & 0 &0 & 1 &  $\dots$ & 0\\
$\vdots$ &$\vdots$ &  $\vdots$  & \vdots &  & $\vdots$ & $\vdots$ &  $\vdots$ & $\vdots$ &  & $\vdots$ &$\vdots$ & $\vdots$ & $\vdots$ &  & $\vdots$\\
$b_n$: &0 &  0 & 0  & $\dots$ & 1 & 0 & 0 & 0  & $\dots$ & 0 & 0 &0 & 0 &  $\dots$ & 1\\ \hline
$c_1$: &0 &  0 & 0  & $\dots$ & 0 & 1 & 0 & 0  & $\dots$ & 0 & 1 &0 & 0 &  $\dots$ & 0\\
$c_2$: &0 &  0 & 0  & $\dots$ & 0 & 0 & 1 & 0  & $\dots$ & 0 & 0 &1 & 0 &  $\dots$ & 0\\
$c_3$: &0 &  0 & 0  & $\dots$ & 0 & 0 & 0 & 1  & \dots & 0 & 0  &0 & 1 &  $\dots$ & 0\\
$\vdots$ &$\vdots$ &$\vdots$& $\vdots$ &   & $\vdots$ & $\vdots$ &  $\vdots$ & $\vdots$ & & $\vdots$ &$\vdots$ & $\vdots$ & $\vdots$ &  & $\vdots$\\
$c_n$: &0 &  0 & 0  & $\dots$ & 0 & 0 & 0 & 0  & $\dots$ & 1 & 0 &0 & 0 & $\dots$ & 1\\
\end{tabular}
\end{table}
Consider a subset of rows $R=A\cup B \cup C$, where $A\subseteq \{a_1,a_2,a_3\}$, $B\subseteq \{b_1,\dots,b_n\}$, $C\subseteq \{c_1,\dots,c_n\}$. We examine the following cases and for each of them we show a valid coloring.
\begin{enumerate}
	\item $A=\emptyset$ . Then $l(b_i)=1$ for $b_i\in B$ and $l(c_i)=-1$ for $c_i\in C$.
	\item $A=\{a_1\}$. Then $l(a_1)=1$, $l(b_i)=-1$ for $b_i\in B$ and $l(c_i)=-1$ if $c_i\in C$.
	\item $A=\{a_2\}$. Symmetric to Case 2.
	\item $A=\{a_3\}$. Then $l(a_3)=1$, $l(b_i)=-1$ for $b_i\in B$ and $l(c_i)=-1$ if $c_i\in C$.
	\item $A=\{a_1,a_2\}$. Then $l(a_1)=1$, $l(a_2)=-1$, $l(b_i)=-1$ for $b_i\in B$ and $l(c_i)=1$ for $c_i\in C$.
	\item $A=\{a_1,a_3\}$. Then $l(a_1)=1$, $l(a_3)=-1$,  $l(b_i)=-1$ for $b_i\in B$ and $l(c_i)=1$ for $c_i\in C$.
	\item $A=\{a_2,a_3\}$. Symmetric to Case 6.
	\item $A=\{a_1,a_2,a_3\}$. Then $l(a_1)=1$, $l(a_2)=1$, $l(a_3)=-1$, $l(b_i)=-1$ for $b_i\in B$, $l(c_i)=-1$ for $c_i\in C$.
\end{enumerate}
\end{proof}  
 From Theorem~\ref{thmtu1} we 
immediately get that every extreme solution of~(\ref{mip0}), after removing the integrality constraints, is integral and in consequence 
  \textsc{Recoverable Selection} for 
  the interval uncertainty representation 
  is polynomially solvable. 
  Our goal is now to construct an efficient  combinatorial algorithm for this problem. 
  In order to do this, we first apply  Lagrangian relaxation (see, e.g.,~\cite{AMO93}) to~(\ref{mip0}). 
  Relaxing  the cardinality constraint   $\sum_{i\in[n]}z_i \geq p-k$
  with a nonnegative multiplier~$\theta$, we obtain 
  the following linear programming problem:
\begin{equation}
\label{dual1}
\begin{array}{rllllll}
	\phi(\theta)= \min & \sum_{i=1}^n C_i x_i + \sum_{i=1}^n \overline{c}_i y_i + \sum_{i=1}^n (C_i+\overline{c}_i-\theta)z_i +(p-k)\theta\\
	\text{s.t. } & x_1+\dots + x_n+z_1+\dots + z_n=p \\
	& y_1+\dots+y_n + z_1+\dots + z_n =p  \\
	& x_i+z_i \leq 1 & i\in [n] \\
	& y_i+z_i\leq 1 & i\in [n] \\
	& x_i, y_i, z_i \geq  0 & i\in [n] \\
\end{array}
\end{equation}  
The Lagrangian function~$\phi(\theta)$  for
any $\theta\geq 0$ is a lower bound on the optimum value~$opt_2$.
It is well-known that $\phi(\theta)$ is concave and piecewise linear function.
 We now find a nonnegative  multiplier~$\theta$ together with an optimal solution $x^*_i,y^*_i,z^*_i$, $i\in[n]$, to ~(\ref{dual1})  which is also feasible in~(\ref{mip0}) and satisfies the complementary slackness condition,
  i.e.  $\theta((p-k)-\sum_{i\in[n]}z_i)=0$. By  the optimality test, 
  such a solution is optimal  to the original problem~(\ref{mip0}).  We will do this by iteratively increasing the value of $\theta$, starting with $\theta=0$.
  For $\theta=0$ the following lemma holds:
\begin{lem}
\label{lfsol}
The value of $\phi(0)$ can be computed in $O(n)$ time.
\end{lem}
\begin{proof}
Let $X$ be a set of $p$ items of the smallest values of $C_i$ and let $Y$ be the set of $p$ items
 of the smallest $\overline{c}_i$, $e_i\in E$.
 Clearly, $\hat{c}=\sum_{e_i\in X} C_i+\sum_{e_i\in Y} \overline{c}_i$ is a lower bound on 
 $\phi(0)$. A feasible solution of the cost $\hat{c}$ can be obtained by setting  $z_i=1$ for $e_i\in X\cap Y$, 
 $x_i=1$ for $e_i\in X\setminus Y$ and $y_i=1$ for $e_i\in Y\setminus X$. The sets $X$ and 
 $Y$ can be found in $O(n)$ time (see the comments in Section~\ref{sintro}),
 and the lemma follows.
\end{proof}

Given a  $0-1$ optimal solution to~$(\ref{dual1})$ for a fixed $\theta\geq 0$ (such a solution must exist due to Theorem~\ref{thmtu1}), let 
$E_X=\{e_i\in E: x_i=1\}$, $E_Y=\{e_i\in E: y_i=1\}$, $E_Z=\{e_i\in E: z_i=1\}$. 
For $\theta>0$ the sets $E_X$, $E_Y$ and $E_Z$ 
are pairwise disjoint and thus form a partition of the set $X\cup Y$ into $X\setminus Y$,
$Y\setminus X$ and $X\cap Y$, respectively.
  Indeed, $E_X\cap E_Z=\emptyset$ and $E_Y\cap E_Z=\emptyset$ by the constraints of~(\ref{dual1}). 
  If $e_i\in E_X\cap E_Y$, then we can find a better solution by setting $z_i=1$, $x_i=0$ and $y_i=0$.
The same property holds for the optimal solution when $\theta=0$ (see the construction in
the proof of Lemma~\ref{lfsol}). Let us state the above reasoning as the following property: 
\begin{pro}
\label{parp}
	For each $\theta\geq 0$ there is an optimal solution $x_i,y_i,z_i$, $i\in [n]$, to~(\ref{dual1}) such that 
	$E_X$, $E_Y$ and $E_Z$  form a partition of the set $X\cup Y$ into $X\setminus Y$,
         $Y\setminus X$ and $X\cap Y$, respectively.
\end{pro}
From now on, we will represent an optimal solution to~(\ref{dual1}) for any $\theta\geq 0$ as  a triple 
$(E_X,E_Y,E_Z)$ that has
Property~\ref{parp}.
 It is easily seen that~(\ref{dual1}) is equivalent to the minimum cost flow problem in the network $G_\theta$ shown in Figure~\ref{net}a. All the arcs in $G_\theta$ have capacities equal to~1. The costs of arcs $(s,x_i)$ and $(y_i,t)$ for $i\in [n]$
 are equal to 0. The costs of the arcs $(x_i,y_i)$ for $i\in [n]$ are $C_i+\overline{c}_i-\theta$, and the costs of the arcs $(x_i,y_j)$, $i\neq j$, are $C_i+\overline{c}_j$. The supply at the source~$t$ is of~$p$ and the demand at
 the sink~$t$ is of $-p$.
 
 \begin{figure}[h]
\centering
\includegraphics{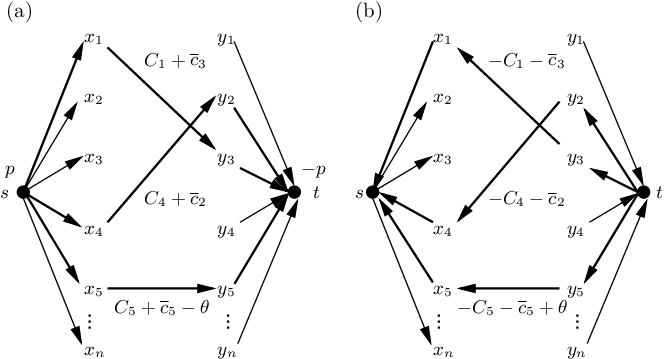}
\caption{(a) Network $G_\theta$ (all the arc capacities are equal to~1, not all arcs between $x_i$ and $y_j$ are shown). (b)  Residual network $G^r_\theta$ for $E_X=\{e_1,e_4\}, E_Y=\{e_2,e_3\}, E_Z=\{e_5\}$ (the arcs leading from $x_i$ to $y_j$ are not shown).} \label{net}
\end{figure}

  Let $(E_X,E_Y,E_Z)$
  be an optimal solution to~(\ref{dual1}). The corresponding integer optimal flow $f_\theta$ in $G_\theta$ is constructed as follows. We send 1 unit of flow along arcs $(s,x_i)$, $(x_i,y_i)$ and $(y_i,t)$ if $e_i\in E_Z$; we then pair the items from $e_i\in E_X$ and $e_j\in E_Y$ in any fashion and send 1 unit  of flow along the arcs 
  $(s,x_i)$, $(x_i,y_j)$ and $(y_j,t)$ for each such a pair $(e_i,e_j)$. An example for 
  $E_Z=\{e_5\}$, $E_X=\{e_1,e_4\}$ and $E_Y=\{e_2,e_3\}$ is shown in Figure~\ref{net}a, where
  $X=\{e_1,e_4,e_5\}$, $Y=\{e_2,e_3,e_5\}$, $p=3$.
   Assume that $f_\theta$ is an optimal  flow in $G_\theta$. We can assume that this flow is integer by 
   the  integrality property  of optimal solutions to the minimum cost flow problem (see, e.g.,~\cite{AMO93}). Let $A\subseteq [n]$ be the set of indices of all nodes
   $x_i$ which receive 1 unit of flow from $s$ and let $B \subseteq [n]$ be the set of indices of nodes $y_j$ which send 1 unit of flow to $t$ in $f_\theta$. Clearly $|A|=|B|=p$. 
   Let $E_Z=\{e_i\in E : i\in A\cap B\}$, $E_X=\{e_i\in E : i\in A\setminus B\}$, 
   $E_Y=\{e_i\in E : i\in B\setminus A\}$.
    It is easy to see that the cost of the resulting feasible solution
    $(E_X,E_Y,E_Z)$
     to~(\ref{dual1}) is the same as the cost of $f_\theta$.

Consider now an optimal solution $(E_X,E_Y,E_Z)$ to~(\ref{dual1}) for some fixed $\theta\geq 0$. 
Without loss  of generality we can assume that this optimal solution has  Property~\ref{parp}.
 Let $f_\theta$ be the corresponding optimal flow in $G_\theta$. The residual network 
 $G^r_\theta$ with respect to $f_\theta$ is depicted in Figure~\ref{net}b. By the negative cycle optimality condition (see, e.g.,~\cite{AMO93}), 
  $G^r_\theta$ does not contain any negative cost directed cycle. Suppose we increase $\theta$ in $G^r_\theta$. Then, a negative cost directed cycle may appear in  $G^r_{\theta}$, which means that the flow $f_\theta$ becomes not optimal in $G_{\theta}$. We now investigate the structure of such negative cycles. This will allow us to find the largest value of $\theta$ for which the flow $f_{\theta}$ (and  thus the corresponding solution to~(\ref{dual1})) remains optimal. 

Denote by $\mathcal{F}$ the set of arcs of the form $(x_i,y_i)$ in $G^r_{\theta}$ (the \emph{forward} arcs) and by 
$\mathcal{B}$ the set of arcs of the form $(y_i,x_i)$ in $G^r_{\theta}$ (the \emph{backward} arcs). These arcs will play a crucial role as only their costs depend on $\theta$ in $G^r_\theta$. Clearly, the costs of the arcs in 
$\mathcal{F}$ decrease and the costs of the arcs in $\mathcal{B}$ increase when the value of~$\theta$ increases.  In the example shown in Figure~\ref{net}b the arc $(y_5,x_5)$ is a backward arc and all the remaining arcs $(x_i,y_i)$ are forward arcs. 

\begin{figure}[h]
\centering
\includegraphics{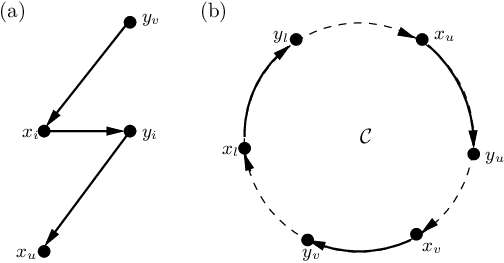}
\caption{(a) A path which cannot appear in $G^r_\theta$. (b) A cycle $\mathcal{C}$ in  $G^r_\theta$ which
  contains  three arcs from $\mathcal{F}$
  (the dashed arcs represent paths between nodes that may traverse $s$ or $t$).} \label{cycle}
\end{figure}

We start by establishing the following  lemma: 
\begin{lem}
\label{obseg}
The residual network $G^r_\theta$ does not contain a path  composed of arcs $(y_v,x_i)$, $(x_i,y_i)$ and $(y_i,x_u)$, where $i,u,v\in [n]$ and $(x_i,y_i)\in \mathcal{F}$.
\end{lem}
\begin{proof}
If such a path exists (see Figure~\ref{cycle}a), then the flow~$f_\theta$ corresponds to  solution $(E_X,E_Y,E_Z)$ in which $x_i=1$ and $y_i=1$. In consequence $E_X\cap E_Y\neq \emptyset$ which violates Property~\ref{parp}.
\end{proof}

The following lemma is crucial in investigating the structure of cycles in $G^r_\theta$.
\begin{lem}
	\label{lsc}
	Each simple cycle in $G^r_\theta$ contains at most two arcs from $\mathcal{F}$.
\end{lem}
\begin{proof}
Assume, contrary to our claim, that
 there exists a cycle $\mathcal{C}$ in  $G^r_\theta$ that
  contains at least three arcs from $\mathcal{F}$ (see Figure~\ref{cycle}b). This cycle may (or may not) contain nodes $s$ and $t$. However, in all possible four cases the cycle $\mathcal{C}$ must violate Lemma~\ref{obseg}, i.e. it must contain a path of the form shown in Figure~\ref{cycle}a.
\end{proof}

Suppose  we increase $\theta$ to some value $\theta'>\theta$ in $G^r_{\theta}$ and denote the resulting residual network as $G^r_{\theta'}$. Assume that a negative cycle $\mathcal{C}$ appears in $G^r_{\theta'}$. It is obvious that
 $\mathcal{C}$ 
 must contain at least one arc from $\mathcal{F}$, since only the costs of these arcs decrease
  when $\theta$ increases. By Lemma~\ref{lsc}, the cycle $\mathcal{C}$ 
  contains either one or two arcs from $\mathcal{F}$.
  
  Consider first the case when $\mathcal{C}$ contains exactly one arc from 
  $\mathcal{F}$, say $(x_u,y_u)$. 
  Then $\mathcal{C}$ 
  cannot contain any arc from $\mathcal{B}$. Otherwise, when computing the cost of $\mathcal{C}$, the value of $\theta'$ would 
   be canceled
   and $\mathcal{C}$ is a negative cycle in $G^r_\theta$, a contradiction. 
  All the possible cases are shown in Figure~\ref{allcyc}. Notice that in the case~(b),
  by Lemma~\ref{obseg},
   the arc $(y_u,t)$ 
  must belong to $\mathcal{C}$.
The dashed arcs represent paths between nodes, say $y_a$ and $x_b$, that do not use any arcs from $\mathcal{F}$ and 
$\mathcal{B}$ and none of the nodes $s$ or $t$. 	
An easy computation shows that the cost of such paths equals $-C_a-\overline{c}_b$ as the remaining terms are canceled while traversing this path.

\begin{figure}[h]
\centering
\includegraphics{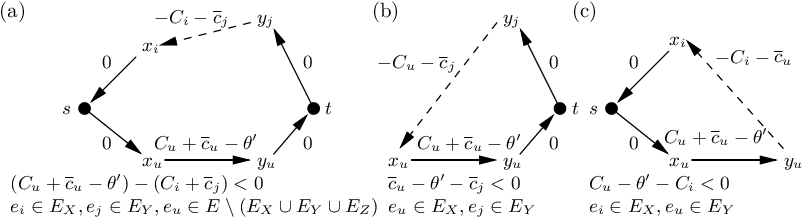}
\caption{All possible situations when $\mathcal{C}$ contains only one arc from $\mathcal{F}$.} \label{allcyc}
\end{figure}

We now turn to the case when $\mathcal{C}$ contains two arcs from 
$\mathcal{F}$, say $(x_l,y_l)$ and $(x_m,y_m)$. It 
is easy to check that in
this case $\mathcal{C}$ may also contain an arc from $\mathcal{B}$, but at most one such arc. The cycle 
$\mathcal{C}$ must traverse $s$ and $t$.
Accordingly,
 it  is  of the form presented in Figure~\ref{sumcyc}a,
 which is a consequence of Lemma~\ref{obseg} and Lemma~\ref{lsc}.
 \begin{figure}[h]
\centering
\includegraphics{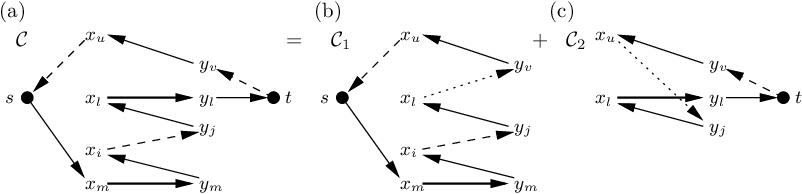}
\caption{A situation when $\mathcal{C}$ contains two arcs from $\mathcal{F}$ and $u\neq v$.} \label{sumcyc}
\end{figure}
Assume that a path from $t$ to $s$ in $\mathcal{C}$ uses  arc $(y_v,x_u)$ such that $u\neq v$ 
(see Figure~\ref{sumcyc}a).
 Let us add to $\mathcal{C}$
  two arcs $(x_u,y_j)$ and $(x_l,y_v)$ whose costs are $C_u+\overline{c}_j$ and 
  $C_l+\overline{c}_v$, respectively
   (see Figure~\ref{sumcyc}a).
   A trivial verification shows that
    the cost of the cycle $x_u \rightarrow y_j \rightarrow x_l \rightarrow y_v \rightarrow x_u$ is 0.
  Hence the cost of $\mathcal{C}$ is the sum of the costs of two disjoint cycles $\mathcal{C}_1$ 
  and $\mathcal{C}_2$ (see Figures~\ref{sumcyc}b and~\ref{sumcyc}c).
   Since the cost of $\mathcal{C}$ is negative, 
   the cost of at least one of  $\mathcal{C}_1$ and  $\mathcal{C}_2$ must be negative.
    If no arc from $\mathcal{B}$ belongs to $\mathcal{C}$, then the cycle $\mathcal{C}_1$ is of the form depicted in
     Figure~\ref{allcyc}c and $\mathcal{C}_2$ is of the form given in Figure~\ref{allcyc}b. 
      Now suppose that $\mathcal{C}$ contains one arc from $\mathcal{B}$. If that arc is in 
       path from $x_i$ to $y_j$ or from $x_u$ to $s$, then it belongs to cycle $\mathcal{C}_1$.
        In this case $\mathcal{C}_2$ is negative and is of the form shown in Figure~\ref{allcyc}b. 
        If that arc is in path from $t$ to $y_v$, then we get a symmetric case.  
      Consequently, the case $u\neq v$, can be reduced to the cases presented in Figure~\ref{allcyc}.
The last case of cycle $\mathcal{C}$, which cannot be reduced to any cases previously discussed, and thus must be treated separately, is shown in Figure~\ref{cyc2}. Observe that in this case contains two arcs from 
$\mathcal{F}$ and one arc from $\mathcal{B}$.
\begin{figure}[h]
\centering
\includegraphics{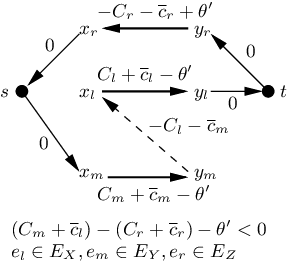}
\caption{A situation when $\mathcal{C}$ contains two arcs from 
$\mathcal{F}$ and one arc from $\mathcal{B}$.} \label{cyc2}
\end{figure}

We are thus led to the following optimality conditions.
\begin{lem}
\label{toptcon}
	Let $(E_X,E_Y,E_Z)$ be an optimal solution to $(\ref{dual1})$ for 
	$\theta\geq 0$. Then $(E_X,E_Y,E_Z)$  remains optimal for all $\theta'>\theta$ 
	which satisfy the following inequalities:
	\begin{align}
	  &C_u+\overline{c}_u - C_i-\overline{c}_j \geq \theta' &\text{for }&
         e_i\in E_X, e_j \in E_Y, e_u\in E \setminus (E_X\cup E_Y\cup E_Z), \label{tcon1}\\
         &\overline{c}_u-\overline{c}_j \geq \theta'  &\text{for } &e_u\in E_X, e_j \in E_Y,  \label{tcon2}\\
         &C_u-C_i \geq \theta'  &\text{for } & e_i\in E_X, e_u \in E_Y, \label{tcon3}\\
         &C_m+\overline{c}_l - (C_r+\overline{c}_r) \geq \theta'  &\text{for } &
         e_l\in E_X, e_m \in E_Y, e_r \in E_Z. \label{tcon4}
	\end{align}
\end{lem}
\begin{proof}
	Consider the optimal flow $f_\theta$ associated with $(E_X,E_Y,E_Z)$ 
	and the corresponding residual network $G^r_{\theta}$. 
	Fix $\theta'>\theta$ so that the conditions (\ref{tcon1})-(\ref{tcon4}) 
	are satisfied and suppose that $f_\theta$ is not optimal in $G_{\theta'}$. 
	Then a negative cost directed cycle must appear in $G^r_{\theta'}$ with respect to $f_\theta$. 
	This cycle must be of the form depicted in Figures~\ref{allcyc} and~\ref{cyc2}. 
	In consequence, at least one of the conditions (\ref{tcon1})-(\ref{tcon4})  must be violated.
\end{proof}

Assume that $(E_X,E_Y,E_Z)$ is an optimal solution to~(\ref{dual1}) 
for some $\theta\geq 0$. This solution must satisfy the optimality conditions (\ref{tcon1})-(\ref{tcon4})  (when $\theta'=\theta$). Suppose that one of the inequalities  (\ref{tcon1})-(\ref{tcon4}) is binding, i.e. it is satisfied as equality.
 In this case we can construct  a new solution $(E^{'}_X,E^{'}_Y,E^{'}_Z)$ whose cost is the same as $(E_X,E_Y,E_Z)$
  by using the following transformations: 
\begin{enumerate}
	\item If  (\ref{tcon1}) is binding, then $E^{'}_Z=E_Z\cup\{e_u\}$, $E^{'}_X=E_X\setminus\{e_i\}$,
	 $E^{'}_Y=E_Y\setminus \{e_j\}$,
	\item If  (\ref{tcon2}) is binding, then  $E^{'}_Z=E_Z\cup\{e_u\}$, 
	$E^{'}_X=E_X\setminus\{e_u\}$, $E^{'}_Y=E_Y\setminus \{e_j\}$,
	\item if  (\ref{tcon3}) is binding, then $E^{'}_Z=E_Z\cup\{e_u\}$, 
	$E^{'}_X=E_X\setminus\{e_i\}$, $E^{'}_Y=E_Y\setminus \{e_u\}$,
	\item if   (\ref{tcon4}) is binding, then $E^{'}_Z=E_Z\cup\{e_l,e_m\}\setminus \{e_r\}$, 
	$E^{'}_Y=E_Y\setminus\{e_m\}$, $E^{'}_X=E_X\setminus\{e_l\}$.
\end{enumerate}
It is worth pointing out that the 
above transformations correspond to augmenting a flow around the corresponding cycle in 
$G^r_\theta$ whose cost is 0. In consequence the cost of $(E^{'}_X,E^{'}_Y,E^{'}_Z)$
 is the same as the cost of $(E_X,E_Y,E_Z)$. Hence $(E^{'}_X,E^{'}_Y,E^{'}_Z)$ is 
 also optimal to~(\ref{dual1}) for $\theta$. Furthermore it must satisfy  conditions (\ref{tcon1})-(\ref{tcon4}) 
 (otherwise we could decrease the cost of $(E^{'}_X,E^{'}_Y,E^{'}_Z)$
  by augmenting a flow around some negative cycle). 
  Observe that $|E^{'}_Z|=|E_Z|+1$ and $(E^{'}_X,E^{'}_Y,E^{'}_Z)$ also has Property~\ref{parp}.
    If none of the inequalities (\ref{tcon1})-(\ref{tcon4})  is binding, then we increase $\theta$ until 
    at least one inequality among (\ref{tcon1})-(\ref{tcon4})
     becomes binding, preserving the optimality of $(E_X,E_Y,E_Z)$. 
       We start with $\theta=0$ and we repeat the procedure described  until 
     we find a optimal solution  $(E_X,E_Y,E_Z)$ to~(\ref{dual1}) for $\theta$,  which is feasible 
      in~(\ref{mip0}) and satisfies the complementary slackness condition:
      $\theta((p-k)-|E_Z|)=0$. Notice that $|E_Z|=p-k$, when $\theta>0$.
      Such a solution is an optimal one to the original problem~(\ref{mip0}).
     Indeed,
     \begin{align}
     opt_2\geq\phi(\theta)&=
     \sum_{e_i\in E_X} C_i + \sum_{e_i\in E_Y} \overline{c}_i 
      + \sum_{e_i\in E_Z} (C_i+\overline{c}_i-\theta) +(p-k)\theta \label{bpe1}\\
      &=
       \sum_{e_i\in E_X} C_i + \sum_{e_i\in E_Y} \overline{c}_i+ \sum_{e_i\in E_Z} (C_i+\overline{c}_i).  \label{bpe2}
     \end{align}
     The inequality~(\ref{bpe1}) follows from the Lagrangian bounding principle
     (see, e.g.,~\cite{AMO93}). Applying  the complementary slackness condition yields the equality~(\ref{bpe2}).
     Moreover $(E_X,E_Y,E_Z)$ is feasible in~(\ref{mip0}),  and so it is optimal solution to~(\ref{mip0}).
     We have thus arrived  to the following lemma.
   \begin{lem}
	The problem~(\ref{mpsp}) is solvable in $O((p-k+1)n^2)$ time.
	\label{tonests}
\end{lem}     
\begin{proof}
By Lemma~\ref{lfsol}
the first solution for $\theta=0$ can be computed in $O(n)$ time. 
Given an optimal solution for some $\theta\geq 0$, the next optimal solution can be found in $O(n^2)$ times, since we need to analyze $O(n^2)$ inequalities  (\ref{tcon1})-(\ref{tcon4}). 
The cardinality of $|E_Z|$
increases by 1 at each step until $|E_Z|=p-k$. Hence, the overall running time of the algorithm is $O((p-k+1)n^2)$.
\end{proof}
From Lemma~\ref{tonests}  and (\ref{irs})  we immediately get the following theorem. 
\begin{thm}
	For scenario set $\mathcal{U}^I$, the \textsc{Recoverable Selection} problem is solvable in $O((p-k+1)n^2)$ time.
\end{thm}

\section{Robust two stage selection}
\label{sec2st}

In this section we explore the complexity of \textsc{Two-Stage Selection}. 
We begin with  a result on the problem under scenario set $\mathcal{U}^I$.
\begin{thm}
	Under scenario set $\mathcal{U}^I$, \textsc{Two-Stage Selection} is solvable in $O(n)$ time.
\end{thm}
\begin{proof}
	Define  $\hat{c}_e=\min\{C_e, \overline{c}_e\}$, $e\in E$. Let $Z$ be the set of $p$ items of the smallest values of $\hat{c}_e$.  Clearly, $\hat{c}=\sum_{e \in Z} \hat{c}_e$ is a lower bound on $opt_1$. A solution $X$ with the cost $\hat{c}$ can be now constructed as follows. For each $e \in Z$, if $C_e\leq \overline{c}_e$, then we add $e $ to $X$; otherwise we add $e $ to $Y$. It holds $f_1(X,S)=\sum_{e\in X} C_e + \sum_{e\in Y} \overline{c}_e=\hat{c}$ and $X$ must be optimal. It is easy to see that $X$ can be computed in $O(n)$ time.
\end{proof}
The problem under consideration is much harder for scenario set $\mathcal{U}^D$. The following theorems hold.

\begin{thm}
	Under scenario set $\mathcal{U}^D$, \textsc{Two-Stage Selection} is  NP-hard  when the number of scenarios equals two.
\label{thm2st1}
\end{thm}
\begin{proof}
In order to prove this theorem we will adapt the idea from~\cite{BBI14}. Consider the following \textsc{Subset Sum} problem, which is known to be NP-hard~\cite{GJ79}. We are given a collection of nonnegative integers $\mathcal{A}=\{a_1,\dots, a_n\}$ and a nonnegative
 integer~$b$, $b\leq \sum_{i\in [n]} a_i$. We ask if there is a subset $I\subseteq [n]$ such that $\sum_{i\in I} a_i=b$. Given an instance $(\mathcal{A},b)$ of \textsc{Subset Sum} we build the corresponding instance of \textsc{Two-Stage Selection} in the following way. We create a set of $n+1$ items $E=\{e_1,\dots,e_{n+1}\}$. Let $M=2\sum_{i\in [n]} a_i$. The first stage costs of all the items in $E$ are set to $M$. For each $i\in [n]$, the second stage cost of $e_i$ under $S_1$ is $a_i+M$ and under $S_2$ it is equal to $M-a_i$. The second stage cost of the item $e_{n+1}$ is equal to $M-2b$ under $S_1$ and $M$ under $S_2$. We complete the reduction by fixing $p=n+1$. Observe that all the items must be selected and the problem consists in determining a subset $X\subseteq E$ of the items chosen in the first stage (then all the items in $E\setminus X$ must be selected in the second stage under each scenario). The reduction is presented in Table~\ref{tab2st}.
\begin{table}[h]
	\caption{The instance of \textsc{Two-Stage Selection} with $p=n+1$, corresponding to \textsc{Subset Sum}.} \label{tab2st}
	\centering
	\begin{tabular}{c|c|cccc}
	$E$	 & $C_e$  & $S_1$ & $S_2$ \\ \hline
	$e_1$ & $M$ & $M+a_1$ & $M-a_1$  \\
	$e_2$ & $M$ & $M+a_2$ & $M-a_2$  \\
	$\vdots$ & $\vdots$ & $\vdots$&$\vdots$ \\
	$e_n$ & $M$ & $M+a_n$ & $M-a_n$  \\ \hline
	$e_{n+1}$ & $M$ & $M-2b$ & $M$ 
	\end{tabular}
	\end{table}

We now show that the answer to \textsc{Subset Sum} is yes if and only if there is a solution $X$ whose cost $F(X)$ is not greater than $(n+1)M-b$. 
By the construction, such a solution~$X$ does not contain the item~$e_{n+1}$ ($e_{n+1}$ is not chosen in the first stage). Indeed, if $e_{n+1}\in X$, then the cost of $X$ under $S_1$ is not less than $(n+1)M$.
Thus we can assume that $e_{n+1}\notin X$. Define by $I$ the set of indices of the items from $X$ (note that $I\subseteq [n]$) and let $I'=[n]\setminus I$. Since $p=n+1$, the cost of $X$ can be expressed as follows:
$$F(X)=\sum_{e\in X} C_e+\max\left\{\sum_{e\in E\setminus X} c^{S_1}_e, \sum_{e\in E\setminus X} c^{S_2}_e\right\}.$$
By the definition of the first and second stage costs we get
\begin{align*}
F(X)&=M|I|+\max\left\{\sum_{i\in I'} (M+a_i)+ M-2b, \sum_{i\in I'}  (M-a_i)+M\right\} \\
&=M(n+1)+\max\left\{\sum_{i\in I'} a_i -2b, \sum_{i\in I'} -a_i\right\}\\
&=M(n+1)-b+\max\left\{\sum_{i\in I'} a_i-b, b-\sum_{i\in I'} a_i\right\}\\
&=M(n+1)-b+|\sum_{i\in I'} a_i -b|.
\end{align*}
Now it is easily seen that $F(X)\leq M(n+1)-b$ if and only if $\sum_{i\in I'\subseteq [n]} a_i=b$, i.e. when the answer to \textsc{Subset Sum} is yes.
\end{proof}

\begin{thm}
	\label{thmhard2}
	When $K$ is a part of input, then \textsc{Two-Stage Selection}  under scenario set $\mathcal{U}^D$ is  strongly NP-hard. Furthermore, 
there exists a constant $c>0$, such that  the problem
is NP-hard to  approximate within a factor of $c\ln n$.	
\end{thm}
\begin{proof}
Consider the following \textsc{Min-Set Cover} problem. We are given a finite set 
$U=\{u_1,\dots,u_n\}$, called the \emph{universe},  and a collection $\mathcal{A}=\{A_1,\dots,A_m\}$ of subsets of $U$. A subset $\mathcal{D}\subseteq \mathcal{A}$ covers~$U$ (it is called a \emph{cover}) if for each element $u_i\in U$, there exists $A_j\in \mathcal{D}$ such that $u_i\in A_j$. We seek a cover $\mathcal{D}$ of the smallest size $|\mathcal{D}|$. 
 \textsc{Min-Set Cover} is known to be strongly NP-hard (see, e.g,~\cite{GJ79}) and
 there exists a constant $c>0$ such that  it
is NP-hard to  approximate within a factor of $c\ln n$~\cite{RS97}. 
 We now show a cost preserving reduction from \textsc{Min-Set Cover} to \textsc{Two-Stage Selection}.
Given an instance $(U,\mathcal{A})$ of \textsc{Min-Set Cover} we construct the corresponding instance of \textsc{Two-Stage Selection} as follows. For each set $A\in \mathcal{A}$ we create an item $e_{A}$ with first-stage cost equal to~1. We also create additional $m$ items labeled as $g_1,\dots,g_m$ with the first-stage costs equal to a sufficiently large number $M$, say $M=|U||\mathcal{A}|$. 
Thus $E=\bigcup_{A\in\mathcal{A}}\{e_A\}\cup \bigcup_{i\in [m]}\{g_i\}$, $|E|=2m$.
The scenario set $\mathcal{U}^D$ is constructed in the following way. For each element $u_i\in U$ we form scenario $S_{u_i}$ under which the cost of $e_A=M$ if $u_i \in A$ and 0 otherwise,
$|\mathcal{U}^D|=|U|$.
 Let $r$ be the number of $M$'s created under scenario $S_{u_i}$ for the items~$e_A$, $A\in \mathcal{A}$. Hence $r$ is the number of sets in $\mathcal{A}$ which contain~$u_i$. We set the costs of $g_1,\dots, g_r$ equal to~0 and the costs of $g_{r+1},\dots,g_m$ equal to $M$. Notice that under each scenario $S_{u_i}$ exactly $m$ items have costs equal to~0. We fix $p=m+1$.
An example of the reduction is depicted in Table~\ref{tab1nn}.
	\begin{table}[h]
	  \centering
	  \caption{The instance of \textsc{Two-Stage Selection} with $p=7$ for $U=\{u_1,\dots,u_7\}, \mathcal{A}=\{\{u_2,u_4,u_3\}, \{u_1\}, \{u_3,u_7\}, \{u_1,u_4,u_6,u_7\}, \{u_2,u_5,u_6\}, \{u_1,u_6\}\}$} \label{tab1nn}
			\begin{tabular}{l|c|cccccccc}
								$E$		 & $C_e$ & $S_{u_1}$ & $S_{u_2}$ & $S_{u_3}$ &  $S_{u_4}$ & $S_{u_5}$ & $S_{u_6}$ & $S_{u_7}$ & \\ \hline
					$e_{\{u_2,u_4,u_3\}}$ 			& 1  & 0 & $M$ & $M$ & $M$ & 0 & 0& 0 \\
					$e_{\{u_1\}}$  				& 1 & $M$ & 0& 0 & 0 & 0& 0 & 0 \\  
					$e_{\{u_3, u_7\}}$ 				& 1 & 0 & 0 & $M$ & 0 & 0& 0 & $M$ \\
					$e_{\{u_1,u_4,u_6,u_7\}}$ 			& 1 & $M$ & 0 & 0 & $M$ & 0& $M$& $M$\\  
					$e_{\{u_2,u_5,u_6\}}$				& 1 & 0 & $M$ & 0 & 0& $M$& $M$ & 0\\
					$e_{\{u_1,u_6\}}$ 				& 1 & $M$ & 0 & 0 & 0 & 0& $M$ & 0  \\  \hline
					$g_1$ 				& $M$ & $0$ & $0$   & $0$ & $0$& $0$& $0$& 0\\
					$g_2$ 				& $M$ & $0$ & $0$   & 0 & $0$ & $M$& $0$& 0 \\ 
					$g_3$				& $M$ & $0$ & $M$  & $M$ & $M$ & $M$ & $0$ & $M$  \\
					$g_4$				& $M$ & $M$ & $M$ & $M$ & $M$ & $M$ & $M$ & $M$\\
					$g_5$ 				& $M$ & $M$ & $M$  & $M$ & $M$ & $M$ & $M$ & $M$\\
					$g_6$				& $M$ & $M$ & $M$  & $M$ & $M$ & $M$ & $M$ & $M$
			\end{tabular}
		\end{table}

	We now show that there is a cover $\mathcal{D}$ of size~$s$ if and only if there
	 is a solution $X$ such that $f_1(X,S)=s$ for every $S\in \mathcal{U}^D$. 
	 Let $\mathcal{D}$ be a cover of size~$s$. In the first stage we choose the items $e_A$ for each $A\in \mathcal{D}$; let $X$ be this set, $|X|=|\mathcal{D}|$.
	  Consider any scenario $S_{u_i}$. Let $E'=E\setminus X$. Since the element $u_i$ is covered, there must exist at least $m-|X|+1$ elements in $E'$ with 0 costs under $S_{u_i}$. 
	  We use these elements in the second stage
	  and form set $Y^{S_{u_i}}$ so that $|X\cup Y^{S_{u_i}}|=p$ , which gives $f_1(X,S_{u_i})=|\mathcal{D}|$ and, consequently, $f_1(X,S_{u_i})=s$
	  for every $S_{u_i}\in \mathcal{U}^D$. Assume now that
	  there
	 is a solution $X$ such that $f_1(X,S)=s$ for every $S\in \mathcal{U}^D$.
	  By the construction and the fact that $s<M$, 
	  $|X|=s$ and $X$ contains only the items corresponding to~$\mathcal{A}$.
	  Let $\mathcal{D}=\{A: e_A\in X\}$, $|\mathcal{D}|=|X|$. Consider scenario  $S_{u_i}$. 
	  Since $f_1(X,S_{u_i})=|X|=s$,
	  we must be able to form set $Y^{S_{u_i}}$, such that $|X\cup Y^{S_{u_i}}|=p$,
	  with $m-|X|+1$ items of 0 cost under $S_{u_i}$. This is possible only when the cost of some $e_A\in X$ under $S_{u_i}$ is $M$, i.e. when the element $u_i$ is covered by $\mathcal{D}$. In consequence, each element $u_i$ is covered by $\mathcal{D}$ and
	  $\mathcal{D}$  is of size~$s$.
It is clear that the presented reduction is cost preserving and the theorem follows.
\end{proof}

We now present a positive result for \textsc{Two-Stage Selection} under scenario set~$\mathcal{U}^D$. Namely,
we  construct an LP-based randomized 
approximation algorithm for this problem, which returns an $O(\log K+ \log n)$-approximate solution with high probability. Consider the following linear program:
\begin{align}
\mathcal{LP}(L):&\sum_{e\in E} C_e x_e+c^{S}_e y^{S}_e\leq L,& &S\in \mathcal{U}^D,\label{sc}\\
                          &\sum_{e \in E}x_e+y^S_e=p,&&S\in \mathcal{U}^D, \label{bc2}\\
                          &x_e+y^{S}_e\leq 1, &e\in E,\; &S\in \mathcal{U}^D, \label{cc}\\
		        &x_e=0,  &e\notin E^1(L), \;&  \label{hc3}\\
		        &y^{S}_e=0,  &e\notin E^{S}(L),\; &S\in \mathcal{U}^D, \label{hc4}\\
		         &x_e\geq 0,  &e\in E^1(L),\;& \label{hc1} \\
                         &y^{S}_e\geq 0,  &e\in E^{S}(L),\; &S\in \mathcal{U}^D,\label{hc2} 
\end{align}
where $E^1(L)=\{e\in E\,:\, C_e\leq L\}$ and $E^S(L)=\{e\in E\,:\, c^S_e\leq L\}$.
Minimizing $L$ subject to~(\ref{sc})-(\ref{hc2}), we obtain an LP relaxation of \textsc{Two-Stage Selection}.
Let  $L^*$ denote  the smallest value of the parameter~$L$
for which $\mathcal{LP}(L)$ is feasible.
Obviously, $L^*$ is a lower bound on $opt_1$, and can be determined in polynomial time by using binary search.  
We lose nothing by assuming, from now on, that 
 $L^*=1$, and all the item costs  are such that~$C_e, c^{S}_e\in[0,1]$, $e\in E$, $S\in \mathcal{U}^D$.
 One can easily meet this
assumption by dividing  all the item costs by $L^*$.  
Notice that we can assume that $L^*>0$. Otherwise, when $L^*=0$,
there exists an optimal integral solution of the zero total cost. Such a solution can be constructed by picking all the items with
zero first and second stage costs under all scenarios.
 
 \begin{algorithm}
  \caption{A randomized algorithm for \textsc{Two-Stage Selection}.}\label{alg1}
 \begin{algorithmic}[1]
 \STATE  $c_{\max}:= \max_{e\in E}\{C_e, 
\max_{S\in \mathcal{U}} c_e^S\}$
\STATE  Use 
binary search in $[0,(n-1)c_{\max}]$ to find 
 the minimal value of $L^*$ for which there exists
 a feasible solution $(\hat{\pmb{x}},(\hat{\pmb{y}}^S)_{S\in \mathcal{U}^D})$ to $\mathcal{LP}(L^*)$\\
 \COMMENT{\textbf{Randomized rounding}} 
\STATE $\hat{t}:=\lceil 32\ln n +8\ln (2K) \rceil$
\STATE $X:=\emptyset$, $Y^S:=\emptyset$ for each $S\in\mathcal{U}^D$
\STATE For each $e\in E$, flip an $\hat{x}_e$-coin $\hat{t}$ times, if it comes up heads at least once
             add $e$ to $X$\label{k5}
\STATE For each $e\in E$ and each $S\in \mathcal{U}^D$, flip an $\hat{y}^S_e$-coin  $\hat{t}$ times, if it comes up heads at least once add item $e$ to $Y^S$\label{k6}\\
 \COMMENT{\textbf{End of randomized rounding}} 
\STATE Add to $X$  at most~4 arbitrary items, which have been not selected in steps~\ref{k5} and~\ref{k6}\label{k7}
\STATE \textbf{If} $|X\cup Y^S|\geq p$ for each $S\in \mathcal{U}^D$ \textbf{then} \textbf{return} $X, Y^S, S\in \mathcal{U}^D$ \textbf{else fail}
\end{algorithmic}
\end{algorithm}

Now our aim is to  convert  a feasible  solution~$(\hat{\pmb{x}},\hat{\pmb{y}})=(\hat{\pmb{x}},(\hat{\pmb{y}}^S)_{S\in \mathcal{U}^D})\in [0,1]^{n+nK}$
 to  $\mathcal{LP}(L^*)$  into a  feasible solution 
  of \textsc{Two-Stage Selection}. Let $x$-\emph{coin} be a coin which comes up head with probability $x\in [0,1]$. We use such a device to construct a randomized  algorithm (see Algorithm~\ref{alg1}) for the problem.
If Algorithm~\ref{alg1} outputs a solution such that $|X\cup Y^S|\geq p$ for each $S\in \mathcal{U}^D$, then the sets $X$ and $Y^S$ can be converted  into a feasible solution  in the following way.
 For each scenario $S\in \mathcal{U}^D$,  If $e\in X\cap Y^S$, then we remove $e$ from $Y^S$. Next, if $|X \cup Y^S|> p$, then we remove arbitrary items, first from $Y^S$ and then from $X$ so that $|X\cup Y^S|=p$. Notice that this operation does not increase the total cost of the selected items under any scenario. The algorithm fails if $|X\cup Y^S|<p$ for at least one scenario $S$. We will show, however, that this bad event occurs with a small probability.  Let us first analyze the cost of the obtained solution.
 
 \begin{lem}
Fix scenario $S\in \mathcal{U}^D$. 
  The probability that the total cost of the items selected in 
 the first and the second stage under
$S$,
  after   the randomized rounding,
 is at  least $\hat{t}L^*+(\mathrm{e}-1)\sqrt{\hat{t}L^{*}\ln(2Kn^2)}$
  is at most $\frac{1}{2Kn^2}$, where $\hat{t}=\lceil 32\ln n +8\ln (2K) \rceil$.
  \label{lbou1}
\end{lem}
\begin{proof}
Fix scenario $S\in \mathcal{U}^D$.
Let $\mathrm{X}_e$ be a random variable such that $\mathrm{X}_e=1$ if item~$e$ is included 
in $X$;
 and $\mathrm{X}_e=0$ otherwise, and let 
$\mathrm{Y}^S_e$ be a random variable such that $\mathrm{Y}^S_e=1$ if item~$e$ is included 
in $Y^S$;
 and $\mathrm{Y}^S_e=0$ otherwise.
Obviously, $\mathrm{Pr}[\mathrm{X}_e=1]=1-(1-\hat{x}_e)^{\hat{t}}$ and
$\mathrm{Pr}[\mathrm{Y}^S_e=1]=1-(1-\hat{y}^S_e)^{\hat{t}}$.
 Because $\hat{t}\geq 1$ and 
$\hat{x}_e,\hat{y}^S_e\in [0,1]$,
an easy computation shows that
 $\mathrm{Pr}[\mathrm{X}_e=1]\leq \hat{x}_e\hat{t}$ and  $\mathrm{Pr}[\mathrm{Y}^S_e=1]\leq \hat{y}^S_e\hat{t}$.
 Define $\mathrm{C}^S=\sum_{e\in E} C_e \mathrm{X}_e+c^{S}_e \mathrm{Y}^S_e$ and let $\psi^S$ be the event that $\mathrm{C}^S>\hat{t}L^*+(\mathrm{e}-1)\sqrt{\hat{t}L^*\ln(2Kn^2)}$, we recall that $L^*=1$.
 The following inequality holds:
 \begin{align}
 \mathbf{E}[\mathrm{C}^S]&=\sum_{e\in E} C_e\mathrm{Pr}[\mathrm{X}_e=1] +
             c^{S}_e \mathrm{Pr}[\mathrm{Y}^S_e=1]=\sum_{e\in E} C_e(1-(1-\hat{x}_e)^{\hat{t}}) +
             c^{S}_e(1-(1-\hat{y}^S_e)^{\hat{t}})\nonumber\\
            & \leq \hat{t}\sum_{e\in E} C_e\hat{x}_e+c^{S}_e\hat{y}^S_e\leq \hat{t}L^*.\label{expcoub}
 \end{align}
 The item costs are such that $C_e, c^S_e\in [0,1]$, $e\in E$. Thus using~(\ref{expcoub})
 and applying Chernoff-Hoeffding bound given in
  \cite[Theorem~1 and inequality (1.13) for $D(\hat{t}L^*,1/2Kn^2)$]{R88}, we obtain
\begin{equation}
 \label{eub1}
 \mathrm{Pr}\left[\psi^S\right]\leq \mathrm{Pr}\left[ \mathrm{C}^S> 
  \hat{t}L^*+(\mathrm{e}-1)\sqrt{\hat{t}L^*\ln(2Kn^2)}\right]< \frac{1}{2Kn^2},
 \end{equation}
 which completes the proof.
 \end{proof}

We now analyze the feasibility of the obtained solution. In order to do this, it is convenient 
to see steps~\ref{k5} and~\ref{k6} of the algorithm in the following equivalent way. In a \emph{round} we flip an $\hat{x}$-coin for each $e\in E$ and add $e$ to $X$ when it comes up head; we then flip an $\hat{y}^S$-coin for each $e\in E$ and $S\in \mathcal{U}^D$ and add $e$ to $Y^S$ if it comes up head. Clearly, steps~\ref{k5} and~\ref{k6} can be seen as performing $\hat{t}$ such rounds independently. Let us fix scenario $S\in \mathcal{U}^D$.
Let $X_t$ and $Y^S_t$ be the sets of items selected in the first and second stage under $S$ (i.e. added to $X$ and $Y^S$), respectively, after $t$ rounds.  Define $E^S_t=E\setminus (X_t\cup Y^S_t)$ and $|E^S_t|=N^S_t$. Initially, $E^{S}_0=E$, $N^S_0=n$. 
Let $P^S_t$ denote the number of items remaining for selection out of the set $E^{S}_t$  under scenario~$S$ 
after the $t$th  round. Initially $P^S_0=p$.
We say that a round~$t$ is ``successful''  if
either $P^S_{t-1}<5$ (at most~4 items are to be selected) or
$P^S_{t}<0.88P^S_{t-1}$;
otherwise, it is ``failure''.
\begin{lem}
Fix scenario $S\in \mathcal{U}^D$. The conditional probability that round~$t$ is  
``successful'', given any set of items~$E^{S}_{t-1}$ 
and number~$P^S_{t-1}$,
is at least $1/2$.
\label{lalon}
\end{lem}
\begin{proof}
If  $P^S_{t-1}<5$, then 
we are done.  Assume that $P^S_{t-1}\geq 5$ and
consider the set of items~$E^{S}_{t-1}$, $|E^{S}_{t-1}|=N^{S}_{t-1}$ and
the number of items $P^S_{t-1}$, remaining for selection in round $t$ (i.e. after round $t-1$).
Let $\mathrm{I}_e$ be a random variable such  
that $\mathrm{I}_e=1$ if item~$e$ is  picked from $E_{t-1}^S$; 
and $\mathrm{I}_e=0$, otherwise. It is easily seen that
$\mathrm{Pr} [\mathrm{I}_e=1]=1-(1-\hat{x}_e) (1-\hat{y}^S_e)$.
The expected number of items selected out of~$E^{S}_{t-1}$ in round $t$
is
\begin{align*}
\mathbf{E}\left[\sum_{e\in E^{S}_{t-1}} \mathrm{I}_e\right]&=
\sum_{e\in E^{S}_{t-1}}\mathrm{Pr}\left[\mathrm{I}_e=1\right]=
N^{S}_{t-1}-\sum_{e\in E^{S}_{t-1}}(1-\hat{x}_e) (1-\hat{y}^S_e)\\
&\geq N^{S}_{t-1}-\sum_{e\in E^{S}_{t-1}}\left(1-\frac{\hat{x}_e+\hat{y}^S_e}{2} \right)^2\geq
N^{S}_{t-1}-\sum_{e\in E^{S}_{t-1}}\left(1-\frac{\hat{x}_e+\hat{y}^S_e}{2} \right)\geq
\frac{P^S_{t-1}}{2}.
\end{align*}
The first inequality follows from the fact that $ab\leq(a+b)/2$ for any $a,b\in [0,1]$ (indeed, $0\leq(a-b)^2=a^2-2ab+b^2\leq a-2ab+b$).
The last inequality follows from the fact that  the feasible  solution~$(\hat{\pmb{x}},\hat{\pmb{y}})$
satisfies  constraints~(\ref{bc2}).
Using  Chernoff bound (see, e.g.,\cite{MR95},
Theorem~4.2 and inequality (4.6) for $\delta=\sqrt{4\ln 2/P^S_{t-1}}$), we get
\[
\mathrm{Pr}\left[\sum_{e\in E^{S}_{t-1}} \mathrm{I}_e < \frac{P^S_{t-1}}{2}-\sqrt{P^S_{t-1}\ln 2}\right]\leq \frac{1}{2}.
\]
Thus, with  probability at 
least $1/2$, the number of selected items in round $t$  is at least
$P^S_{t-1}/2-\sqrt{P^S_{t-1}\ln 2}$.
 Hence, with probability at least $1/2$ it holds
\[
P^S_t\leq P^S_{t-1}-P^S_{t-1}/2+\sqrt{P^S_{t-1}\ln 2}=P^S_{t-1}(1/2+\sqrt{\ln 2/P^S_{t-1}}).
\]
Consequently, when $P^S_{t-1}\geq 5$ we get
$$P^S_t<0.88 P^S_{t-1}$$
with probability at least $1/2$.
\end{proof}
 \begin{lem}
Fix scenario $S\in \mathcal{U}^D$.
 The probability of the event that $P^S_{\hat{t}}\geq 5$  is at most $1/(2Kn^2)$, where
 $\hat{t}=\lceil 32\ln n +8\ln (2K) \rceil$.
  \label{lbou2}
 \end{lem}
\begin{proof}
Let $\xi^S$ be the event that  $ P^S_{\hat{t}}\geq 5$, i.e. that the number of items remaining for selection after $\hat{t}=\lceil 32\ln n +8\ln 2K \rceil$ rounds is at least~5. 
We now estimate the number $\ell$ of successful rounds which are enough to achieve 
$P^S_{\ell}<5$.  It is easy to see that $\ell$  satisfies $(0.88)^{\ell} p \leq 4$.
 In particular, this inequality holds when $\ell\geq 8\ln p$.
Let~$\mathrm{Z}$ be a random variable denoting the number
of  ``successful''
rounds among $\hat{t}$~performed rounds.
We estimate $\mathrm{Pr}[\mathrm{Z}<8\ln p]$ from above by 
$\mathrm{Pr}[\mathrm{B}(\hat{t},1/2)<8\ln p]$, where
$\mathrm{B}(\hat{t},1/2)$ is a binomial random variable.
This 
can be done, since we have 
a lower bound on  success of given any history.
Applying  Chernoff bound 
(see, e.g.,
\cite[Theorem~4.2 and inequality (4.6) for $\delta=\ln (2Kn^2)/\ln (2Kn^4)$]{MR95} and
$\hat{t}=\lceil 32\ln n +8\ln (2K) \rceil$)
 and the fact that $p<n$,
 we obtain the following upper bound:
\begin{align}
\mathrm{Pr}[\xi^S]\leq \mathrm{Pr}[\mathrm{Z}<8\ln p]&\leq 
\mathrm{Pr}[\mathrm{B}(\hat{t},1/2)<8\ln p]
\leq \mathrm{e}^{-(8\ln n+4\ln (2K))^2/2(16\ln n+4\ln(2K))}\nonumber\\
&<\mathrm{e}^{-(8\ln n+4\ln (2K))/4}= \frac{1}{2Kn^2},\label{eub2}
\end{align}
and the lemma follows.
\end{proof}
  Lemma~\ref{lbou1} and Lemma~\ref{lbou2} (see also inequalities~(\ref{eub1}) and~(\ref{eub2})) 
  and the union bound imply that $\Pr[\xi^{S_1}\cup\dots\cup\xi^{S_K}\cup\psi^{S_1}\dots\cup\psi^{S_K}]< 1/n^2$. 
  Therefore after $\hat{t}=\lceil 32\ln n +8\ln (2K) \rceil$ rounds the cost of solution  $X \cup Y^S$
  found is $O((\log K+ \log n)L^*)$ and 
  the number of items remaining for selection is at most~$4$
  for every $S\in \mathcal{U}^D$
  with probability at least $1-\frac{1}{n^2}$.
The addition of~4 items to $X$ in step~\ref{k7} can  increase the cost of the computed solution by 
at most $4\cdot L^*$. As all the costs are nonnegative, repairing  $X \cup Y^S$, $S\in \mathcal{U}^D$,
   to obtain a feasible solution cannot increase the cost of the computed solution.
   Since $L^*\leq opt_1$, we
   get the following result:
\begin{thm}
There is a polynomial time randomized algorithm for
 \textsc{Two-Stage Selection}
that returns an $O(\log K + \log n)$-approximate solution with high probability.
\end{thm}
It is worth pointing out
that if $K=poly(n)$, then our randomized algorithm gives the best approximation ratio up to a constant factor
(see Theorem~\ref{thmhard2}).

\section{Conclusions and open problems}

In this paper  we have discussed two robust versions of the \textsc{Selection} problem, which have a two-stage nature. In the first problem, a partial solution is formed in the first stage and completed optimally when a true state of the world reveals. In the second problem a complete solution must be formed in the first stage, but it can be modified to some extent after a true state of the world becomes known. Such two-stage problems often appear in practical applications. In this paper we have presented some positive and negative complexity results for two types of uncertainty representations, namely the discrete and interval ones. In particular, we have shown that both problems are polynomially solvable
  for the interval uncertainty representation. We believe that a similar method might be applied to other combinatorial optimization problems, in particular for those possessing a matroidal structure. When the number of scenarios is a part of input, then  the recoverable model is not at all approximable and the two-stage model 
  has an approximability lower bound of~$\Omega(\log n)$. 
   We have shown that the latter one admits a randomized $O(\log n+ \log K)$-approximation algorithm.

There are still some open questions concerning the considered problems. A deterministic $\log n$-approximation algorithm for the two-stage problem may exist and it can be a subject of further research. When $K$ is constant, then both robust problems are only proven to be weakly NP-hard. So, they might be solved in pseudopolynomial time and even admit an FPTAS. The interval uncertainty representation can be generalized by adopting the scenario set proposed in~\cite{BS03}. The complexity of both robust problems under this scenario set is open.

\subsubsection*{Acknowledgements}
This work was  supported by
 the National Center for Science (Narodowe Centrum Nauki), grant  2013/09/B/ST6/01525.


\end{document}